\documentclass{article}
\usepackage{amsthm}
\usepackage{amssymb}
\usepackage{epsfig}
\usepackage[dvips]{color}

\newcount\claimno
\def\mytextindent#1{\indent\indent\llap{\rm#1\enspace}\ignorespaces}
\def\myitem{\par\hangindent30pt\mytextindent}
\def\myclaim#1#2{
   \global\advance\claimno by 1\relax
   \bigskip\noindent\rlap{\rm(\the\claimno)}\ignorespaces
   \global\expandafter\edef\csname CLAIMLABEL#1\endcsname{(\the\claimno)}\relax
   \hangindent=33pt\hskip30pt{\sl#2}\bigskip}
\def\refclaim#1{\csname CLAIMLABEL#1\endcsname}
\def\endclaim{\medskip}
\def\junk#1{}
\let\ppar=\par
\def\rt#1{#1}

\def\mylabel#1{{\label{#1}}}
\newtheorem{theorem}{Theorem}

\begin{document}
\title{Three-coloring triangle-free graphs on surfaces I. 
       Extending a coloring to a disk with one triangle}
\author{%
     Zden\v{e}k Dvo\v{r}\'ak\thanks{Computer Science Institute,
           Charles University, 
           Malostransk{\'e} n{\'a}m{\v e}st{\'\i} 25, 118 00 Prague, 
           Czech Republic. E-mail: {\tt rakdver@iuuk.mff.cuni.cz}.
           Supported by the Center of Excellence -- Inst. for Theor. Comp. Sci., Prague (project P202/12/G061 of Czech Science Foundation).}
 \and
     Daniel Kr{\'a}l'\thanks{Mathematics Institute, DIMAP and Department of Computer Science, University of Warwick, Coventry CV4 7AL. E-mail: {\tt D.Kral@warwick.ac.uk}.}
 \and
        Robin Thomas\thanks{School of Mathematics, 
        Georgia Institute of Technology, Atlanta, GA 30332. 
        E-mail: {\tt thomas@math.gatech.edu}.
        Partially supported by NSF Grant No.~DMS-0701077.}
}
\date{March 4, 2016}
\maketitle
\begin{abstract}
Let $G$ be a plane graph with exactly one triangle $T$ and all 
other cycles of length at least $5$, and let $C$ be a facial cycle of $G$
of length at most six.
We prove that a $3$-coloring of $C$ does not extend to a $3$-coloring of $G$ 
if and only if $C$ has length exactly six and there is a color $x$ such that
either $G$ has an edge joining two vertices of $C$ colored $x$,
or $T$ is disjoint from $C$ and every vertex of $T$ is adjacent to a vertex
of $C$ colored $x$.
This is a lemma to be used in a future paper of this series.
\end{abstract}

\section{Introduction}

This is the first paper in a series aimed at studying the $3$-colorability
of graphs on a fixed surface that are either triangle-free, or have their
triangles restricted in some way.
All graphs in this paper are simple, with no loops or parallel edges.

The subject of coloring graphs on surfaces goes back to 1890 and the work
of Heawood~\cite{Heawood}, who proved that if $\Sigma$ is not the
sphere, then every graph in $\Sigma$
is $t$-colorable as long as
$t\ge H(\Sigma):=\lfloor(7+\sqrt{24\gamma+1})/2\rfloor$.
Here and later $\gamma$ is the {\em Euler genus of $\Sigma$},
defined as $\gamma=2g$ when $\Sigma=S_g$, the orientable surface
of genus $g$, and $\gamma=k$ when $\Sigma=N_k$, the non-orientable surface
with $k$ cross-caps.
Incidentally, the assertion holds for the sphere as well, by the
Four-Color Theorem~\cite{AppHak1,AppHakKoc,AppHak89,RobSanSeyTho4CT}.
Ringel and Youngs (see~\cite{Ringel}) proved that the bound is
best possible for all surfaces except the Klein bottle, for which
the correct bound is $6$.
Dirac~\cite{Dirmap} and Albertson and Hutchinson~\cite{AlbHut}
improved Heawood's result by showing that
every graph in $\Sigma$ is actually $(H(\Sigma)-1)$-colorable,
unless it has a subgraph isomorphic to the complete graph
on $H(\Sigma)$ vertices.

For triangle-free graphs there does not seem to be a similarly nice formula,
but Gimbel and Thomassen~\cite{gimbel} gave very good bounds: they
proved that the maximum chromatic number of a triangle-free graph drawn
in a surface of Euler genus $\gamma$ is at least $c_1 \gamma^{1/3}/\log \gamma$
and at most $c_2(\gamma/\log \gamma)^{1/3}$ for some absolute constants 
$c_1$ and $c_2$.

In this series we adopt a more modern approach to coloring graphs on surfaces,
following the seminal work of 
Thomassen~\cite{Tho5torus,ThoCritical,thom-surf}.
The basic premise is that while Heawood's formula is best possible
for all surfaces except the Klein bottle, only relatively few graphs attain
the bound or even come close.
To make this assertion more precise let us recall that a graph $G$ is
called {\em $k$-critical}, where $k\ge1$ is an integer, if every proper
subgraph of $G$ is $(k-1)$-colorable, but $G$ itself is not.
It follows easily from Euler's formula that if $\Sigma$ is a fixed surface,
and $G$ is a sufficiently big graph drawn in $\Sigma$, then $G$ has 
a vertex of degree at most six. 
It follows that for every $k\ge8$ the graph $G$ is not $k$-critical,
and hence there
are only finitely many $k$-critical graphs that can be drawn in $\Sigma$.
It is not too hard to extend this result to $k=7$.
In fact, it can be extended to $k=6$ by the following deep theorem
of Thomassen~\cite{ThoCritical}.

\begin{theorem}
\mylabel{thm:thomcrit}
For every surface $\Sigma$ there are only finitely many $6$-critical 
graphs that can be drawn in $\Sigma$.
\end{theorem}

\noindent
The lists of $6$-critical graphs are explicitly known for the
projective plane~\cite{AlbHut}, the torus~\cite{Tho5torus} and the 
Klein bottle~\cite{ChePosStrThoYer,KawKraKynLid}.
An immediate consequence is that for every surface $\Sigma$ there is
a polynomial-time (in fact, linear-time) algorithm to test whether
an input graph drawn in $\Sigma$ is $5$-colorable.
Theorem~\ref{thm:thomcrit} does not hold for $5$-critical graphs,
because of an elegant construction of Fisk~\cite{Fis}.
For $3$-colorability an algorithm as above does not exist, unless $P=NP$,
because testing $3$-colorability is NP-hard even for planar 
graphs~\cite{GarJoh}.
It is an open problem whether there is a polynomial-time algorithm for testing
$4$-colorability of graphs in $\Sigma$ when $\Sigma$ is a fixed
surface other than \rt{the} sphere.
The techniques currently available do not give much hope for a positive
resolution in the near future. 

How about triangle-free graphs?
Similarly as above, if $G$ is a sufficiently large triangle-free graph
in a fixed surface $\Sigma$, then $G$ has a vertex of degree at most four.
Thus $G$ is not $6$-critical, and the argument can be strengthened to show
that $G$ is not $5$-critical.
Thus for a fixed integer $k\ge4$ testing $k$-colorability of triangle-free graphs
drawn in a fixed surface can be done in linear time, as before.
That brings us to testing $3$-colorability of triangle-free graphs
in a fixed surface, the subject of this series of papers.
The question has been raised by
Gimbel and Thomassen~\cite{gimbel} and we resolve it later
in this series, after we develop some necessary theory.

Historically the first result in this direction is the following
classical theorem of Gr\"otzsch~\cite{grotzsch}.

\begin{theorem}
\label{grotzsch}
Every triangle-free planar graph is $3$-colorable.
\end{theorem}

Thomassen \cite{thom-torus,Tho3list,ThoShortlist}
found three reasonably simple proofs, and extended Theorem~\ref{grotzsch}
to other surfaces.  
Recently, two of us, in joint work with Kawarabayashi~\cite{DvoKawTho}
were able to design a linear-time algorithm to $3$-color triangle-free
planar graphs, and as a by-product found perhaps a yet simpler proof
of Theorem~\ref{grotzsch}.
The statement of Theorem~\ref{grotzsch}
cannot be extended to any surface other than the sphere.
In fact, for every non-planar surface $\Sigma$ there are infinitely many
$4$-critical triangle-free graphs that can be drawn in $\Sigma$.
For instance, the graphs obtained from an odd cycle of length five or more
by applying Mycielski's
construction \cite[Section~8.5]{BonMur} have that property.
Thus an algorithm for testing $3$-colorability of triangle-free graphs
on a fixed surface will have to involve more than just testing the
presence of finitely many obstructions.

The situation is different for graphs of girth at least five
by another deep theorem of Thomassen~\cite{thom-surf}.

\begin{theorem}
\mylabel{thm:thomgirth5}
For every surface $\Sigma$ there are only finitely many $4$-critical
graphs of girth at least five that can be drawn in $\Sigma$.
\end{theorem}

Thus the $3$-colorability problem on a fixed surface
has a polynomial-time algorithm
for graphs of girth at least five, but the presence of cycles of
length four complicates matters.
Let us remark that there are no $4$-critical graphs of girth at least five
on the projective plane and the torus~\cite{thom-torus} and
on the Klein bottle~\cite{thomwalls}.

The only non-planar surface for which the $3$-colorability problem
for triangle-free graphs is fully characterized is the projective plane.
Building on earlier work of Youngs~\cite{Youngs}, Gimbel and
Thomassen~\cite{gimbel} obtained the following elegant characterization.
A graph drawn in a surface is a {\em quadrangulation} if every face
is bounded by a cycle of length four.

\begin{theorem}
\mylabel{thm:gimtho}
A triangle-free graph drawn in the projective plane is $3$-colorable if and only
if it has no subgraph isomorphic to a non-bipartite
quadrangulation of the projective plane.
\end{theorem}

For other surfaces there does not seem to be a similarly nice characterization,
but in a later paper of this series we will present a polynomial-time
algorithm to decide whether a triangle-free graph in a fixed surface
is $3$-colorable.
The algorithms naturally breaks into two steps.
The first is when the graph is
a quadrangulation, except perhaps for a bounded number of larger faces
of bounded size, which will be allowed to be precolored.
In this case there is a simple topological obstruction to the existence
of a coloring extension based on the so-called ``winding number" of
the precoloring.
Conversely, if the obstruction is not present and the graph is highly
``locally planar", then we can show that the precoloring can be
extended to a $3$-coloring of the entire graph.
This can be exploited to design a polynomial-time algorithm.
With additional effort the algorithm can be made to run in linear time.

The second step covers the remaining case, when the graph has either many faces
of size at least five, or one large face, and the same holds for every
subgraph.
In that case we show that the graph is $3$-colorable.
That is a consequence of the following theorem, which will form the
cornerstone of this series.

\begin{theorem}
\mylabel{thm:corner}
There exists an absolute constant $K$ with the following property.
Let $G$ be a  graph drawn in a surface $\Sigma$ of Euler genus $\gamma$
with no non-contractible cycles of length at most four,  and let $t$ be
the number of triangles in $G$.
If $G$ is $4$-critical, 
then $\sum|f|\le K(t+\gamma)$,
where the summation is over all faces $f$ of $G$ of length at least five.
\end{theorem}

If $G$ has girth at least five, then $t=0$ and every face has length
at least five.
Thus Theorem~\ref{thm:corner} implies Theorem~\ref{thm:thomgirth5},
and, in fact, improves the bound given by the proof of 
Theorem~\ref{thm:thomgirth5} in~\cite{thom-surf}.
The fact that our bound in Theorem~\ref{thm:corner} is linear in the
number of triangles is needed in our solution of a problem of 
Havel~\cite{conj-havel}, as follows.

\begin{theorem}
\mylabel{havel}
There exists an absolute constant $d$ such that if $G$ is a planar
graph and every two distinct triangles in $G$ are at distance at least $d$,
then $G$ is $3$-colorable.
\end{theorem}

To prove Theorem~\ref{thm:corner} we actually prove a stronger result.
If $G$ is a graph and $R$ is a subgraph of $G$, then we say that $G$
is {\em $R$-critical} if $G\ne R$ and for every proper subgraph $G'$ of
$G$ that includes $R$ there exists a $3$-coloring of $R$ that extends to
a $3$-coloring of $G'$, but does not extend to a $3$-coloring of $G$.
The stronger version of Theorem~\ref{thm:corner} applies to $R$-critical
graphs, where $R$ has bounded size.
A special case that we need in order to carry out our inductive argument
is the following.

\begin{theorem}
\mylabel{thm:speccase}
There exists an absolute constant $L$ with the following property.
Let $G$ be a planar graph with two distinct facial cycles $R_1$ and $R_2$,
where $R_1$ has length six and $R_2$ has length four.
Assume that every cycle in $G$ of length at most four separates $R_1$ and 
$R_2$, and that every cycle in $G$ other than $R_2$ of length exactly four is disjoint
from $R_2$.
If $G$ is $R_1\cup R_2$-critical, then it has at most $L$ vertices.
\end{theorem}

The machinery we develop in the second paper~\cite{proof-druhy} 
of this series can be
applied to prove Theorem~\ref{thm:speccase}, with one notable
exception: when $G$ is disconnected, and the component $G_1$ of $G$
that contains $R_1$ has a face bounded by a triangle
(in which case this face includes the component of $G$ containing $R_2$).
We need to show that the size of $G_1$ is bounded, but we have not been able to do
so using the same methods we use for the rest of the proof of 
Theorem~\ref{thm:corner}.
Instead, we use a different method, which allows us to characterize
the components $G_1$.
Since the proof method is different, we separate it from the other
arguments and present it in this paper.
Thus our main result is as follows.
The two outcomes are illustrated in Figure~\ref{fig-6plus3}.

\begin{figure}
\begin{center}\epsfbox{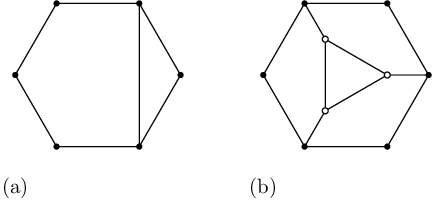}\end{center}
\caption{Critical graphs with a precolored $6$-cycle and one triangle.}
\label{fig-6plus3}
\end{figure}

\begin{theorem}
\mylabel{thm-6plus3}
Let $G$ be a plane graph with a facial cycle $R$ of length at most six,
let $T$ be a triangle in $G$, and assume that
every cycle in $G$ other than $T$ and $R$ has length at least five.
Let $\phi$ be a $3$-coloring of $R$ that does not extend to a $3$-coloring
of $G$. Then $R$ has length exactly six and either
\myitem{(a)}$\phi(u)=\phi(v)$ for two distinct vertices $u,v\in V(R)$ that are
adjacent in $G$, or
\myitem{(b)}$\phi(u_1)=\phi(u_2)=\phi(u_3)$ for three pairwise distinct
vertices $u_1,u_2,u_3\in V(R)$, where each $u_i$ is adjacent to a different
vertex of $T$.
\end{theorem}

Finally, let us mention a related interesting conjecture due to
Steinberg~\cite{conj-stein}, who conjectured that every planar graph 
without $4$- and $5$-cycles is $3$-colorable.
This conjecture is still open. 
Currently the best result of 
Borodin, Glebov, Montassier and Raspaud~\cite{BorGleMonRas}
shows that excluding cycles 
of lengths $4$, $5$ and $7$ suffices to guarantee $3$-colorability. 
(A proof of the same result by Xu~\cite{xu} is refuted in~\cite{BorGleMonRas}.)

\section{Auxiliary results}

In this short section we present several results that will be needed later.
Let $G$ be a graph, and let $R$ be a subgraph of $G$.
\rt{Let us recall} that $G$ is {\em $R$-critical} if $G\ne R$ and 
for every proper subgraph $G'$ of $G$ that includes $R$ as a subgraph
there exists a $3$-coloring of $R$ that extends to a $3$-coloring of $G'$,
but not to one of $G$.
Theorem~\ref{grotzsch} admits the following strengthening.

\begin{theorem}\label{grotzsch5cycle}
There is no $R$-critical triangle-free plane graph $G$, where 
$R$ is a cycle in $G$ of length at most five.
\end{theorem}

This result was later strengthened in several ways.  
Gimbel and Thomassen~\cite{gimbel} extended this to cycles of length six:

\begin{theorem}\label{thm-gimbel}
Let $G$ be a plane triangle-free graph with a facial cycle $R$ of length six. 
If $G$ is $R$-critical, then all faces of $G$ distinct from $R$ have length four.
\end{theorem}

For graphs of girth at least five the above results can be strengthened, 
as shown by Thomassen~\cite{thom-surf} and Walls~\cite{WalPhD}.

\junk{
More can be said about the critical graphs of girth $5$.  
shows the following characterization\footnote{Theorem~\ref{thm-planechar}
was originally formulated in \cite{thom-surf} in a slightly different
setting:  the described graphs have the property that there exists a precoloring of $R$ that does
not extend to $H$, but extends to every proper subgraph of $H$ that contains $R$.  However,
every $\{R\}$-critical graph $G$ has a subgraph $G'$ with this property.  Also, every face of $G'$ of
length at most $7$ is also a face of $G$ (see Lemma~\ref{lemma-crsub} below).
We conclude that $G$ is one of the graphs described in Theorem~\ref{thm-planechar}.}
of critical graphs with precolored face of length at most $12$:
}

\begin{theorem}
\mylabel{thm-planechar}
Let $H$ be a plane graph of girth at least five,
and let $C$ be a facial cycle in $H$ of length  $k\le11$.
If $H$ is $C$-critical, then 
\begin{itemize}
\item[(a)] $k\ge8$, $V(H)=V(C)$ and $C$ is not induced, or
\item[(b)] $k\ge 9$, $H-V(C)$ is a tree with at most $k-8$ vertices, 
           and every vertex of $V(H)-V(C)$ has degree three in $H$, or
\item[(c)] $k\ge 10$ and $H-V(C)$ is a connected graph with at most $k-5$ vertices
           containing exactly one cycle, and the length of this cycle is five.
           Furthermore, every vertex of $V(H)-V(C)$ has degree three in $H$.
\end{itemize}
\end{theorem}

\noindent
{\bf Proof.}
Since $H$ is $C$-critical, there exists a $3$-coloring of $C$ that
does not extend to a $3$-coloring of $G$.
Theorem~2.5 of \cite{thom-surf} (as well as independently proved Theorem~3.0.2 of \cite{WalPhD})
states that there exists a subgraph $H'$ of $H$
such that $C$ is a subgraph of $H'$ and $H'$ satisfies one of (a)--(c).
We claim that $H=H'$. Indeed, otherwise the $C$-criticality of $H$
implies that there exists a $3$-coloring of $C$ that does not extend to
a $3$-coloring of $H$, but extends to a $3$-coloring $\phi$ of $H'$.
By applying~\cite[Theorem~2.5]{thom-surf} or~\cite[Theorem~3.0.2]{WalPhD}  
to every face of $H'$
we deduce that $\phi$ extends to every face of $H'$, and hence extends
to a $3$-coloring of $H$, a contradiction.~\qed
\bigskip

We will need a version of Theorem~\ref{grotzsch5cycle} that allows
the existence of a triangle.
Such a result was attempted by Gr\"unbaum~\cite{grunbaum}, but his proof
is not correct.
A correct proof was found by Aksionov~\cite{aksenov}.  
The result of Aksionov together with Theorem~\ref{thm-gimbel}
gives the following characterization of critical graphs with a triangle 
and a precolored face of length at most five.

\begin{theorem}\label{thm-planetr}
Let $G$ be a plane graph with a facial cycle $R$ of length at most five and at 
most one triangle $T$ distinct from $R$.
If $G$ is $R$-critical, then $R$ has length exactly five, 
$T$ shares at least one edge
with $R$ and all faces of $G$ distinct from $T$ and $R$ have length exactly four.
\end{theorem}

\section{Graphs with one triangle}\label{sec-6plus3}

To prove Theorem~\ref{thm-6plus3} we prove, for the sake of the
inductive argument, the following slightly more general result.
Theorem~\ref{thm-6plus3} will be an immediate corollary.

\begin{theorem}
\mylabel{thm-6plus3new}
Let $G$ be a plane graph with outer cycle $R$ of length at most six
and assume that
\myitem{($*$)}there exists a point $p$ in the plane such that for every cycle $C$ in $G$ 
   of length at most four, the open disk bounded by $C$ contains $p$.\ppar
\noindent
If $G$ is $R$-critical, then 
$R$ has length exactly six
and $G$ is isomorphic to one of the graphs depicted in Figure~\ref{fig-6plus3}.
\end{theorem}

\noindent
{\bf Proof}.
Let $G$ be as stated, and suppose for a contradiction that it 
is not isomorphic to either of the two graphs depicted in
Figure~\ref{fig-6plus3}.
By Theorems~\ref{grotzsch5cycle} and~\ref{thm-gimbel} 
the graph $G$ has a triangle $T$.
We may assume that $G$ is minimal in the sense that the theorem holds
for every graph with fewer vertices.
A vertex $v\in V(G)-V(R)$ will be called {\em internal}. 
The $R$-criticality of $G$ implies that

\myclaim{Z}{every internal vertex of $G$ has degree at least three.}

If $C$ is a cycle in $G$, then by ins$(C)$ we denote the subgraph of $G$
consisting of all vertices and edges drawn in the closed disk bounded by $C$.
More generally, suppose that $C$ is a closed walk in $G$ such that the
subgraph of $G$ consisting of the vertices and edges of $C$ has a
bounded face $\Delta$ incident with all the edges of $C$.  It follows that
$\Delta$ is unique and homeomorphic to an open disk.  In this case, we say that $C$
\emph{bounds the open disk $\Delta$} and we define ins$(C)$ to be the subgraph of $G$
consisting of all vertices and edges drawn in the closure of $\Delta$.


\myclaim{Y}{\rt{Let $C$ be a closed walk in $G$ that bounds an open disk
    $\Delta$, and assume that $\Delta$ includes at least one vertex or
    edge of $G$. Then}
    {\rm ins}$(C)$ is a $C$-critical graph.}

To prove \refclaim{Y} let $C$ be as stated, let 
$G'$ be obtained from $G$ by deleting every vertex and edge of $G$
drawn in $\Delta$, and
let $J$ be a proper subgraph of ins$(C)$ that includes $C$.
Then some $3$-coloring of $R$ extends to a $3$-coloring $\phi$ of
$G'\cup J$, but does not extend to a $3$-coloring of $G$.
It follows that the restriction of $\phi$ to $C$ extends to $J$,
but not to ins$(C)$, as desired.
This proves~\refclaim{Y}.  
\medskip

It follows from ($*$), \refclaim{Y} and Theorem~\ref{thm-planetr} that

\myclaim{B}{$T$ bounds a face and all other cycles in $G$ have length at least five.}

Consequently, $G$ is connected, as otherwise it would contain a component
disjoint from $R$ and this component would be $3$-colorable by Theorem~\ref{thm-planetr},
contrary to the assumption that $G$ is $R$-critical.
Next we constrain cycles in $G$ of length at most seven:

\myclaim{A}{ Let $C\ne R$ be a cycle in $G$ of length at most seven that does 
    not bound a face. Then $C$ has length at least six, and the closed
    disk bounded by $C$ includes $T$.}

To prove \refclaim{A} let $C$ be as stated.
By the minimality of $G$ and~\refclaim{Y} we deduce that $C$ has length
at least six.
If $T$ is not contained in the closed disk $\Delta$ bounded by $C$, then 
\refclaim{B} implies that ins$(C)$ has girth at least five,
contrary to~\refclaim{Y} and Theorem~\ref{thm-planechar}.
Thus $T$ is contained in $\Delta$, and~\refclaim{A} follows.
\endclaim

The same argument and the minimality of $G$ imply the following  claim.

\myclaim{A6}{ Let $C\ne R$ be a cycle in $G$ of length six that does
    not bound a face. Then {\rm ins}$(C)$ is isomorphic to one of the
    graphs depicted in Figure~\ref{fig-6plus3}.}

\myclaim{AA}{Let $C\ne R$ be a closed walk in $G$ of length $k\le11$
  bounding an open disk $\Delta$ disjoint from $T$, and let
  $H=\hbox{\rm ins}(C)$. 
  If $H\ne C$, then $H$ satisfies the conclusion of Theorem~\ref{thm-planechar}.}

To prove~\refclaim{AA} we modify $H$, converting $C$ to a cycle,
and then apply Theorem~\ref{thm-planechar}.
First we replace every edge of $H$ that is traversed by $C$ twice
by a pair of parallel edges bounding a face of length two, thus
creating a multigraph $H_1$ and a closed walk $C_1$ in $H_1$
such that $C_1$ uses no edge more than once.
Thus every vertex of $H_1$ is incident with an even number of edges of $C_1$.
Let $v$ be a vertex of $H_1$ incident with $2k\ge4$ edges of $C_1$.
Then the edges and faces of $C_1$ incident with $v$ can be numbered
$e_1,f_1,e_2,f_2,e_3,\ldots,f_{2k-1},e_{2k},f_{2k}$ in the clockwise cyclic
order of appearance around $v$ in such a way
that for $i=1,2\ldots,k$ the edges $e_{2i-1},e_{2i}$ are consecutive
in $C_1$ and $f_{2i-1}=\Delta$.
We split $v$ into $k$ vertices $v_1,v_2,\ldots,v_k$ as follows.
Let $i=1,2,\ldots,k$, and let $e_{2i-1},g_1,g_2,\ldots,g_l,e_{2i},\ldots$
be the edges of $H$ incident with $v$ listed in clockwise cyclic order
around $v$. Then the edges $e_{2i-1},g_1,g_2,\ldots,g_l,e_{2i}$ will
be incident with $v_i$ in the new graph in the clockwise order listed.
By repeating this construction for every vertex of $H_1$ incident with
at least four edges of $C_1$ we arrive at a plane graph $H_2$ such
that the walk $C_2$ corresponding to $C_1$ is a cycle.
By~\refclaim{Y} the graph $H$ is $C$-critical, and since the above construction
preserves criticality we deduce that $H_2$ is $C_2$-critical.
Thus~\refclaim{AA} follows from Theorem~\ref{thm-planechar}.
\endclaim




From \refclaim{B} and Theorem~\ref{thm-planetr} it follows that

\myclaim{D}{ $R$ has length six,}

\noindent
and since every cycle of length at most five bounds a face by~\refclaim{A}
we deduce that

\myclaim{DD}{ the graph $G$ has no subgraph $H$ with outer face $R$ 
    such that $H$ is isomorphic to either of the two
    graphs depicted in Figure~\ref{fig-6plus3};
    in particular, $R$ is induced.}

Next we claim that

\myclaim{E}{ every internal vertex has at most one neighbor in $V(R)$.}

To prove~\refclaim{E} suppose for a contradiction that
an internal vertex $v_2$ has two neighbors $v_1,v_3\in V(R)$.
Let $P$ denote the path $v_1v_2v_3$, and let $R,C_1,C_2$ be the three
cycles of $R\cup P$.
By~\refclaim{B} either one of $C_1,C_2$ is $T$ and
the other has length seven, or $C_1,C_2$ both have length five.
In either case it follows from~\refclaim{A} that $C_1,C_2$ both bound
faces of $G$, and hence $v_2$ has degree two, contrary to~\refclaim{Z}.
This proves~\refclaim{E}.
\endclaim

\myclaim{F}{ The cycle $T$ is disjoint from $R$.}

To prove~\refclaim{F} suppose for a contradiction that $v\in V(T)\cap V(R)$.
By~\refclaim{DD} and~\refclaim{E} $v$ is the only vertex of $T\cap R$.
The graph $T\cup R$ has a face bounded by a walk $C$ of length nine.
By~\refclaim{AA} combined with \refclaim{B}, at least one of the vertices of
$V(T)\setminus V(R)$ has degree two, which contradicts~\refclaim{Z}.
This proves~\refclaim{F}.
\endclaim

Let us fix an orientation of the plane, and
let $T=t_1t_2t_3$ and $R=r_1r_2\ldots r_6$ be numbered in clockwise cyclic 
order according to the drawing of $G$.

\myclaim{cl-onetr}{ $G$ has at most one edge joining $T$ to $R$.}

To prove~\refclaim{cl-onetr}
suppose that say $t_1r_1, t_2r_i\in E(G)$ for some $i\in\{1,\ldots, 6\}$.
By~\refclaim{B} we have $3\le i\le 5$.
Let $C_2=r_1t_1t_3t_2r_ir_{i+1}\ldots r_6$.  As $t_3$ has degree at least three,
$C_2$ does not bound a face; thus $C_2$ has length at least eight 
by~\refclaim{A}, and we conclude that $i=3$.
Thus $C_2$ has length exactly eight, and hence by~\refclaim{AA}
ins$(C_2)$
consists of $C_2$ and at most one chord.  
Since $t_3$ has degree
at least three, this chord exists and joins $t_3$ with $r_5$, and 
hence $G$ has a subgraph isomorphic to the graph depicted in
Figure~\ref{fig-6plus3}(b), contrary to~\refclaim{DD}.
This proves~\refclaim{cl-onetr}.
\endclaim

\myclaim{cl-no53}{ $G$ does not contain a $5$-face incident only with 
     internal vertices of degree three.}

To prove~\refclaim{cl-no53} suppose
for a contradiction that $G$ contains such a $5$-face $C=v_1v_2v_3v_4v_5$.
For $1\le i\le 5$, let $x_i$ be the neighbor of $v_i$ not belonging to $C$ 
(each $v_i$ has such a neighbor,
because $T$ \rt{bounds a face by~\refclaim{B} and $C$ bounds a face by definition}
and each vertex of $C$ has degree three).  
Since $T$ is disjoint from $R$ by~\refclaim{F} and
$G$ contains no $4$-cycles by~\refclaim{B}, 
it follows that at most three of the vertices 
$x_1, \ldots, x_5$ belong to $R$.  
Without loss of generality we may assume that $x_1$ is internal.
Note also that $x_1\not\in\{ x_3,x_4\}$, as $G$ does not contain a $4$-cycle.
By the symmetry between $x_3$ and $x_4$ we may assume that if
$x_3$ is adjacent to a vertex of $R$, then so is $x_4$.
Let $G'$ be the graph obtained from $G-V(C)$ by adding the edge $x_1x_3$
drawn in the same way as the path $x_1v_1v_2v_3x_3$ in $G$.
Observe that every $3$-coloring of $G'$ extends to a $3$-coloring of $G$: 
given a $3$-coloring of $G'$, every vertex in $C$ has a list of two available
colors, and the lists of $v_1$ and $v_3$ are different.  


Our next objective is to show that $G'$ satisfies ($*$).
To that end let $K'\ne T$ be a cycle in $G'$ of length at most four.
Then $K'$ includes the edge $x_1x_3$ by~\refclaim{B}.
Consider the cycle $K$ in $G$ obtained from $K'$ by replacing the edge
$x_1x_3$ by the path $P=x_1v_1v_2v_3x_3$.
Note that $K$ has length at most seven, and that it does not bound a face 
\rt{of $G$}
(since the edges $v_1v_5$ and $v_2x_2$ are drawn on the opposite sides of $P$).
Thus by~\refclaim{A} $T$ is a subgraph of ins$(K)$, and since the edge $x_1x_3$ of $G'$
is drawn in the same way as the path $P$ of $G$, it follows that the point $p$
is contained in the open disk bounded by $K'$.
We conclude that $G'$ satisfies ($*$), as desired.

\begin{figure}
\begin{center}
\epsfbox{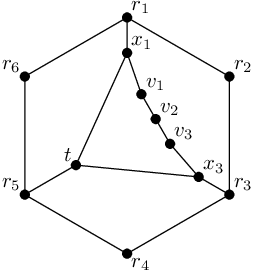}
\end{center}
\caption{A configuration obtained in the proof of \refclaim{cl-no53}.}
\label{fig-12}
\end{figure}

Let $G''$ be a minimal subgraph of $G'$ such that $R$ is a subgraph
of $G''$ and 
every $3$-coloring of $R$ that extends to a $3$-coloring of $G''$ extends
to a $3$-coloring of $G'$.  
Then $G''\ne R$, for otherwise every $3$-coloring of $R$ extends to
a $3$-coloring of $G''$, and hence to a $3$-coloring of $G'$, 
and therefore to one of $G$, contrary to the $R$-criticality of $G$.
We conclude that $G''$ is $R$-critical.
The minimality of $G$ implies that $G''$ is isomorphic to one of the
graphs depicted in Figure~\ref{fig-6plus3}.
But $R$ is an induced subgraph of $G$ by~\refclaim{DD} 
and $x_1$ is internal, and hence $G''$
is isomorphic to the graph of Figure~\ref{fig-6plus3}(b).
Let $L'$ be the triangle of $G''$. 
By \refclaim{cl-onetr} we have $L'\neq T$, 
and hence $x_1x_3$ is an edge of $L'$.
Let $t$ be the third vertex of $L'$.
We may assume that $x_1$ is adjacent to $r_1$, $x_3$ is adjacent to $r_3$
and $t$ is adjacent to $r_5$, where the adjacencies take place in
$G,G'$ and $G''$.
Let $D'$ be the face boundary of the $5$-face of $G''$ incident with the edge 
$x_1x_3$, and let $D$ be the $8$-cycle of $G$ obtained
from $D'$ by replacing the edge $x_1x_3$ by the path $P$ (see Figure~\ref{fig-12}).  
Let $L$ be the $6$-cycle in $G$ obtained from $L'$ by
replacing the edge $x_1x_3$ by the path $P$.  
By~\refclaim{A} $T$ lies in the closed disk bounded by $L$, and 
since $t$ is adjacent to $r_5$ it follows that
ins$(D)$ includes no cycle of length at most four.  
By~\refclaim{AA}
no vertex of $G$ lies in the open disk bounded by $D$, 
and hence $v_4$ and $v_5$ lie in the open disk bounded by $L$.
Since $G$ has no $4$-cycles we deduce that $x_4\not\in\{x_1,t\}$,
and $x_3\ne x_4$, for otherwise $T=x_3v_4v_3$ and the cycle
$x_3tx_1v_1v_5v_4$ includes the edge $v_5x_5$ in its inside but not $T$,
contrary to~\refclaim{A}.
Since $x_3$ is adjacent to $r_3$, the choice of $x_3$ implies that
$x_4$ is adjacent to a vertex of $R$, contrary to the planarity of $G$.
This proves~\refclaim{cl-no53}.
\endclaim




\myclaim{cl-disttr}{ The distance between $R$ and $T$ is at least two.}

To prove~\refclaim{cl-disttr} 
suppose for a contradiction that the distance between $R$ and $T$ is at most one.
Then it is exactly one by~\refclaim{F}, and so we may assume that 
say $t_1r_1\in E(G)$.
\rt{Let $C$ denote the closed walk $r_1r_2\ldots r_6r_1t_1t_2t_3t_1$.
Then $C$ bounds an open disk. Let $H=\hbox{ins}(C)$.
Since $G$ is $R$-critical, we see that $H\ne C$.
By~\refclaim{AA} the graph $H$ satisfies (a), (b) or (c) of
Theorem~\ref{thm-planechar}.}
If it satisfies (a), then by~\refclaim{Z}, $t_2$ has a neighbor in $R$,
contrary to~\refclaim{cl-onetr}.
If $H$ satisfies (b), then
$H-V(C)$ is a tree $X$ with at most three vertices,
each of degree three. Both $t_2$ and $t_3$
have a neighbor in $X$, and hence $X\cup\{t_2,t_3\}$ includes the
vertex-set of a $5$-cycle, contrary to~\refclaim{cl-no53}.
Finally, $H$ cannot satisfy (c), because the cycle referenced
in~(c) would contradict~\refclaim{cl-no53}.
This proves~\refclaim{cl-disttr}.

\myclaim{cl-no22r}{ No two vertices of degree two are adjacent in $G$}.

To prove~\refclaim{cl-no22r}
suppose for a contradiction that $G$ has two adjacent vertices of degree two.
By~\refclaim{Z} they belong to $R$, and so we may assume that 
say $r_2$ and $r_3$ have degree two.  
The edge $r_2r_3$ is not contained in any $5$-cycle,
as otherwise $R$ would have a chord or an internal vertex would have two 
neighbors in $R$, contrary to~\refclaim{DD}
and~\refclaim{E}.
Let $G'$ be the graph obtained from $G$ by contracting the edge $r_2r_3$, 
and let $R'$ be the corresponding outer cycle of $G'$.
Then $G'$ has no cycle of length at most $4$ distinct from $T$.
Furthermore, every $3$-coloring $\psi$ of $R$ can be modified to a 
$3$-coloring $\psi'$ of $R'$ such that $\psi'(r_i)=\psi(r_i)$ for $i\in \{1,4,5,6\}$,
and $\psi$ extends to $G$ if and only if $\psi'$ extends to a $3$-coloring of $G'$. 
It follows that $G'$ is $R'$-critical, contrary to Theorem~\ref{thm-planetr}.  
This proves~\refclaim{cl-no22r}.
\endclaim

\myclaim{cl-no3chord}{ For every path $v_1v_2v_3v_4$ with $v_2$ 
  and $v_3$ internal and $v_1,v_4\in V(R)$
  there exists $r\in V(R)$ such that $v_1v_2v_3v_4r$ bounds a $5$-face.}

To prove~\refclaim{cl-no3chord}
consider a path $P=v_1v_2v_3v_4$ with $v_2$ and $v_3$ internal 
and $v_1,v_4\in V(R)$,
and let $C_1$ and $C_2$ be the cycles of  $R\cup P$ other than $R$ such that $T$
lies in the closed disk bounded by $C_1$.  
Since $C_1$ does not bound a face, it has length at least six by~\refclaim{A}.
Thus $C_2$ has length at most six, and hence bounds a face by~\refclaim{A},
and therefore has length five by~\refclaim{cl-no22r}.
This proves~\refclaim{cl-no3chord}.
\endclaim

\myclaim{cl-faces}{ All faces of $G$ distinct from $R$ and $T$ have 
        length exactly five.} 

To prove~\refclaim{cl-faces}
consider a face $f=v_1v_2\ldots v_k$ of length $k\ge 6$ in $G$.  
By~\refclaim{E} we may assume without loss of generality that
$v_2$ and $v_3$ are internal.  Furthermore, if $k=6$, then not
all of $v_1$, $v_4$, $v_5$ and $v_6$ may belong to $R$, by \refclaim{cl-no22r},
and hence, by symmetry, we may assume that either $v_4$ or $v_6$ is internal.
Let $W=\{v_2,v_k\}$ if $k>6$ and $W=\{v_2,v_4,v_6\}$ if $k=6$.
Let $G'$ be the graph obtained from $G$ by identifying the vertices of $W$ 
to a new vertex $w$
and deleting all resulting parallel edges (the drawing of the new vertex $w$ is placed inside the
face $f$ of $G$ and the drawings of the edges incident with the vertices of $W$ in $G$ are first
shifted infinitesimally so that they do not intersect and then joined to $w$ through $f$). 
Thus $E(G')\subseteq E(G)$.
By~\refclaim{Z} and~\refclaim{A} the vertices of $W$ are pairwise non-adjacent;
thus the identifications created no loops.
Observe that every $3$-coloring $\psi$ of $G'$ gives rise to a $3$-coloring
of $G$ (color the vertices of $W$ using $\psi(w)$).
It follows that some $3$-coloring of $R$ does not extend to a $3$-coloring of $G'$.
Let $G''$ be a minimal subgraph of $G'$ such that $R$ is a subgraph of $G''$
and every $3$-coloring of $R$ that extends to a $3$-coloring of $G''$
also extends to a $3$-coloring of $G'$;
then $G''$ is $R$-critical.

Next we show that $G''$ satisfies ($*$).
As a first step we prove that $G''$ does not have a triangle other than $T$.
To that end let $K'\ne T$ be a triangle in $G''$.
Recall that $E(G'')\subseteq E(G)$.
Two of the edges of $K'$ are incident in $G$ with distinct vertices 
$w_1,w_2\in W$.
Let $K$ be the corresponding $5$-cycle in $G$, obtained from $K'$ by replacing $w$ with
the two-edge path between $w_1$ and $w_2$ with internal vertex in
$\{v_1,v_3,v_5\}$.
Observe that $K$ does not bound a face in $G$, contrary to~\refclaim{A}.
Therefore, $G''$ does not have a triangle distinct from $T$.  
Consider now a $4$-cycle $L'$ in $G''$.
The corresponding cycle $L$ in $G$ (constructed in the same way as $K$)
has length six.  
As $L$ does not bound a face we can apply~\refclaim{A6} to the cycle $L$;
it follows that $T$ is contained in the closed disk bounded by $L$.  Note that the point
$p$ is contained in the open disk bounded by $T$ and since $p$ does not lie inside $f$,
the choice of the drawing of $G'$ ensures that $p$ is contained in the open disk bounded by $L'$.
It follows that $G''$ satisfies ($*$).

Since $G''$ has fewer vertices than $G$,
$G''$ is one of the graphs drawn in Figure~\ref{fig-6plus3}.
Furthermore, the first result of the previous paragraph implies
that $T$ is the unique triangle of $G''$.
However, this implies that the distance between $T$ and $R$ in
$G$ is at most one, contradicting \refclaim{cl-disttr}.
This proves~\refclaim{cl-faces}.
\endclaim

\myclaim{cl-onet4}{At least one vertex of $T$ has degree at least four.}

To prove~\refclaim{cl-onet4} suppose for a contradiction
that all vertices of $T=t_1t_2t_3$ have degree at most three.
By~\refclaim{Z} and \refclaim{F} they have degree exactly three.
For $i=1,\ldots, 3$, let $x_i$ be the neighbor
of $t_i$ that does not belong to $V(T)$.  As $T$ is the only cycle of length
at most four in $G$ \rt{by~\refclaim{B}}, 
these vertices are distinct and pairwise non-adjacent,
and by \refclaim{cl-disttr}, they are internal.

Suppose first that each of $x_1$, $x_2$ and $x_3$ has a neighbor in $R$.
Let $H$ be the subgraph of $G$ 
\rt{consisting of $T\cup R$ and the six edges joining $x_1,x_2,x_3$ to
$T$ and $R$,}
and let $C_1$, $C_2$ and $C_3$ be the cycles bounding the faces of $H$ distinct from $R$ and $T$.
Note that $|V(C_1)|+|V(C_2)|+|V(C_3)|=21$ and 
$G=\hbox{ins}(C_1)\cup\hbox{ins}(C_2)\cup\hbox{ins}(C_3)$.
At most one of $C_1$, $C_2$ and $C_3$ bounds a face of $G$, as otherwise $x_1$, $x_2$ or $x_3$ would be an internal vertex
of degree two, contrary to~\refclaim{Z}.
\rt{From the symmetry we may assume that $C_1$ and $C_2$ do not bound a face.}
From~\refclaim{A} \rt{it follows} that $|V(C_1)|,|V(C_2)|\ge 8$.
This implies that $|V(C_3)|=5$ and $|V(C_1)|=|V(C_2)|=8$.  
However, \refclaim{AA} implies that $C_1$ has a chord, 
contradicting~\refclaim{DD}, \refclaim{E} \rt{or~\refclaim{cl-disttr}}.

Therefore, by symmetry we assume that $x_1$ has no neighbor in $R$.
Let $G'$ be the graph obtained from $G-V(T)$ by adding the edge $x_1x_2$, where the edge is drawn along the
path $x_1t_1t_3t_2x_2$ of $G$.
\rt{Since $G$ is $R$-critical, there exists a $3$-coloring $\phi$ of $R$
that does not extend to a $3$-coloring of $G$.}
Note that every coloring of $G'$ extends to a coloring of $G$, since each vertex of $T$ has a list of two available colors
and the lists of $t_1$ and $t_2$ are different.  
Thus $\phi$ does not extend to  a $3$-coloring of $G'$.
Let $G''$ be a smallest subgraph of $G'$ such that $R$ is a subgraph of
$G''$ and $\phi$ does not extend to $G''$.
Then $G''$ is $R$-critical.

Consider a cycle $K'$ in $G''$ of length at most four.  Note that
$x_1x_2$ is an edge of $K'$, and let $K$ be the cycle obtained from $K'$ by
replacing this edge by the path $x_1t_1t_3t_2x_2$.  The cycle $K$ has length six or seven,
and thus by~\refclaim{AA} and \refclaim{cl-faces}, $T$ is contained in the closed disk bounded by $K$.
By the choice of the drawing of the edge $x_1x_3$, we conclude that the point $p$ is contained in the
open disk bounded by $K'$.  Therefore, $G''$ satisfies ($*$).

As $G''$ has fewer vertices than $G$, $G''$ is isomorphic to one of the graphs in Figure~\ref{fig-6plus3}.
As $T$ is not a subgraph of $G''$, the triangle of $G''$ contains the edge $x_1x_2$, and thus $x_1$
is at distance at most one from $R$ both in $G''$ and in $G$.  This contradicts the choice of $x_1$,
finishing the proof of~\refclaim{cl-onet4}.
\endclaim


\myclaim{cl-excl6}{If $C$ is a $6$-cycle of $G$ distinct from $R$, then the open disk bounded by $C$
contains no vertices and $V(T)\subseteq V(C)$.}

\rt{To prove~\refclaim{cl-excl6} let $C$ be as stated.
By~\refclaim{cl-faces} it does not bound a face, and by~\refclaim{A6} ins$(C)$
is isomorphic to one of the graphs depicted in Figure~\ref{fig-6plus3}.
By~\refclaim{cl-onet4} it is not isomorphic to the second of those two
graphs, and hence~\refclaim{cl-excl6} follows.}
\endclaim

\myclaim{cl-5faces}{ If $f$ is a $5$-face incident with four internal vertices of 
  degree three, then \rt{some edge of $T$ is incident with $f$}.}

To prove~\refclaim{cl-5faces}
suppose for a contradiction that $G$ contains a $5$-face $f=v_1v_2v_3v_4v_5$, 
where $v_1$, $v_3$, $v_4$ and $v_5$ are internal vertices of degree three,
and that \rt{no edge of $T$ is incident with $f$.} 
By \refclaim{cl-no53} the degree of $v_2$ is at least four:
this follows directly from~\refclaim{cl-no53} if $v_2$ is internal;
otherwise $v_2$ has two neighbors in $R$ and two internal neighbors 
\rt{incident with} $f$.
Let $x_1$, $x_3$, $x_4$ and $x_5$ be the neighbors of $v_1$, $v_3$, $v_4$ 
and $v_5$, respectively, outside of $\{v_1,v_2,\ldots,v_5\}$.
If $v_2\in V(R)$, then $x_3$ is internal since $v_3$ has only one neighbor
in $R$ by~\refclaim{E}, and
$x_4$ and $x_5$ are internal by \refclaim{cl-no3chord} and \refclaim{Z}.
Also, not all of $x_1$, $x_3$, $x_4$ and $x_5$ belong to $R$,
as \rt{no edge of $T$ is incident} with $f$.  
Thus we may assume that at least one of $x_3$ and $x_4$ and
at least one of $v_2$ and $x_5$ is internal.  
As $f$ does not share an edge with $T$, the vertices $v_2$, $x_3$, $x_4$ and $x_5$ 
are distinct and pairwise non-adjacent by~\refclaim{Z}, \refclaim{B} and~\refclaim{A}.  
Let $G'$ be the graph obtained from $G-\{v_1,v_3,v_4,v_5\}$ by identifying
$v_2$ with $x_5$ to a new vertex $w_1$ and $x_3$ with $x_4$ to a new vertex $w_2$
(with both $w_1$ and $w_2$ drawn inside the original face $f$ and the edges
incident with $v_2$, $x_5$, $x_3$ and $x_4$ extended towards them in the natural way,
without changing their position with respect to the point $p$).
Note that any coloring $\psi$ of $G'$ extends to a coloring of $G$: 
Give $v_2$ and $x_5$ the color $c_1=\psi(w_1)$
and $x_3$ and $x_4$ the color $c_2=\psi(w_2)$.  
If $c_1=c_2$, then color the vertices of $V(F)\setminus \{v_2\}$
in the order $v_1$, $v_5$, $v_4$ and $v_3$.  
Otherwise, color $v_4$ by $c_1$ and then color $v_3$, $v_1$ and $v_5$ in order.
It follows that some $3$-coloring of $R$ does not extend to a 
$3$-coloring of $G'$.
Let $G''$ be a minimal subgraph of $G'$ such that $R$ is a subgraph of $G''$
and every $3$-coloring of $R$ that extends to a $3$-coloring of $G''$
also extends to a $3$-coloring of $G'$.
Then $G''$ is $R$-critical.

Next we show that $G''$ satisfies ($*$).
Consider  a cycle $K'$ of $G''$ of length at most four distinct from $T$,
and let $K\subseteq G$ be the corresponding cycle obtained by replacing 
$w_1$ by \rt{$v_2$ or $x_5$ or} $v_2v_1v_5x_5$, \rt{as appropriate, and
replacing} $w_2$ by \rt{$x_4$ or $x_3$ or} $x_4v_4v_3x_3$, \rt{as appropriate}.
If we \rt{added both $v_2v_1v_5x_5$ and $x_4v_4v_3x_3$}, 
then $K$ has length at most $10$
and it has two chords $v_2v_3$ and $v_4v_5$.
Thus one of them must belong to $T$, contradicting
the assumption that \rt{no edge of $T$ is incident with $f$.} 
Therefore, we expanded only one vertex in $K'$ \rt{into a path}, 
and hence $6\le |V(K)|\le 7$.  
By \refclaim{cl-faces}, $K$ does not bound a face.
By~\refclaim{A} $T$ is a subgraph of ins$(K)$.  The choice of the drawing of $G'$
ensures that the point $p$ (contained in the open disk bounded by $T$) belongs to
the open disk bounded by $K'$; hence, $G''$ satisfies ($*$), as claimed.

Since $G''$ has fewer vertices than $G$,
we conclude that $G''$ is isomorphic to one of the graphs from 
Figure~\ref{fig-6plus3}.  
Let $K'$ be the unique triangle of $G''$.
Using \refclaim{cl-disttr}, we conclude that $K'\neq T$.  
\rt{From~\refclaim{B}, \refclaim{A} and the fact that no edge of $T$ is 
incident with $f$ we deduce that $w_1$ and $w_2$ are not adjacent in $G'$.
It follows that exactly one of $w_1,w_2$ belongs to $K'$.}
Let $K$ be the corresponding cycle of length six in $G$.
Since $K\neq R$, \refclaim{cl-excl6} implies that $V(T)\subseteq V(K)$.
Let us label the vertices of $K$ so that $K=xvv'x'y_1y_2$,
where $xvv'x'$ is either $v_2v_1v_5x_5$ or $x_3v_3v_4x_4$.
Since \rt{no edge of $T$ is incident with $f$} and $v_1$, $v_3$, $v_4$ and $v_5$
have degree three, it follows that $y_1y_2$ is an edge of $T$.
Let $j\in\{1,2\}$ be the index such that $x$ and $x'$ are identified
into $w_j$.  
\rt{Then $w_{3-j}\not\in \{y_1,y_2\}$, because $w_1$ and $w_2$ are not 
adjacent in $G'$. Thus one of $y_1,y_2$ is not adjacent to $w_{3-j}$
(because $K'$ is the only triangle in $G''$),
and so we may assume that $y_1$ is not adjacent to $w_{3-j}$.}
Since $G''$ is isomorphic to one of the graphs from Figure~\ref{fig-6plus3},
$y_1$ is at distance at most one from $R$ in $G''$, and,
\rt{since $y_1$ is not adjacent to $w_{3-j}$}, we conclude that
$y_1$ is also at distance at most one from $R$ in $G$.
This contradicts~\refclaim{cl-disttr} and proves~\refclaim{cl-5faces}.
\endclaim

\myclaim{cl-no232}{ The cycle $R$ has no subpath $z_1z_2z_3$ 
       with $\deg(z_2)=3$ and $\deg(z_1)=\deg(z_3)=2$.}

To prove~\refclaim{cl-no232} suppose for a contradiction that 
say $r_2$ and $r_4$ have degree two and $r_3$
has degree three.  
By \refclaim{cl-no22r} the vertices $r_1$ and $r_5$ have degree at least three.
By~\refclaim{cl-faces} the face incident with $r_2$ distinct from the 
outer face is bounded by a $5$-cycle, say $r_1r_2r_3yx$.  
Similarly, there is a face bounded by a $5$-cycle $r_3r_4r_5zy$,
where $x\ne z$ by~\refclaim{Z}.
Let $K$ be the $6$-cycle $r_1xyzr_5r_6$.
By~\refclaim{Z} and~\refclaim{A6} the graph ins$(K)$
is isomorphic to one of the graphs in Figure~\ref{fig-6plus3},
contrary to~\refclaim{cl-disttr}.
This proves~\refclaim{cl-no232}.
\endclaim

\myclaim{cl-no233233}{ If $R$ has at least two vertices of degree two, 
   then it has at least one vertex of degree at least four.}

To prove~\refclaim{cl-no233233}
suppose for a contradiction that $R$ has at least two 
vertices of degree  two and the remaining vertices of degree at most three.  
By \refclaim{cl-no22r} and \refclaim{cl-no232} $G$ has exactly two vertices
of degree two, and the distance in $R$ between them is three. 
We may therefore assume that
$r_1$ and $r_4$ have degree two, and $r_2,r_3,r_5,r_6$
have degree three.  
By \refclaim{cl-faces}, $G$ has 
a $6$-cycle $C=x_1x_2x_3x_4x_5x_6$ such that 
$x_1r_2, x_3r_3, x_4r_5, x_6r_6\in E(G)$.
By~\refclaim{A6} the graph ins$(C)$ is isomorphic to
one of the graphs in Figure~\ref{fig-6plus3}.
It follows that either $x_2$ or $x_5$ has degree two, contrary to~\refclaim{Z}.
This proves~\refclaim{cl-no233233}.
\endclaim

We are now ready to complete the proof of Theorem~\ref{thm-6plus3new} using
the so-called discharging argument.
Let us assign charges to the vertices and faces of $G$ in the following way:
Each face $f$ of length $|f|$ not bounded by $R$ or $T$ gets 
a charge of $1=|f|-4$, 
the face bounded by $T$ gets charge $2=(|V(T)|-4)+3$,
and the face bounded by $R$ gets charge $0=(|V(R)|-4)-2$.
A vertex $v\in V(R)$ of degree two gets charge $-1/3=(\deg(v)-4)+5/3$,
a vertex $v\in V(R)$ of degree three gets charge $0=(\deg(v)-4)+1$,
and all other vertices $v$ get charge $\deg(v)-4$.

\myclaim{bubbles}{ The total sum of the charges is at most $-1/3$.}

To prove~\refclaim{bubbles} we deduce from Euler's formula 
the sum of the charges is at most
$\sum_{f\in F(G)}(|f|-4)+\sum_{v\in V(G)} (\deg(v)-4)+n_3+5n_2/3+1=
n_3+5n_2/3-7$, where $n_2$ is the number of vertices of degree two
and $n_3$ is the number of  vertices of $R$ of degree three in $G$. 
By \refclaim{cl-no22r} $n_2\le 3$.  
By \refclaim{cl-no232}, if $n_2=3$ then $n_3=0$.
By \refclaim{cl-no233233}, if $n_2=2$, then $n_3\le 3$.
It follows that $n_3+5n_2/3\le 20/3$, and hence the sum
of the charges is at most $-1/3$, as desired.
This proves~\refclaim{bubbles}.
\endclaim

\junk{
By \refclaim{cl-no22r}, $n_2\le 3$.  
If $n_2\le1$, then $n_3+5n_2/3-7\le-4/3$, and the claim holds.
If $n_2=2$, then either $n_3+5n_2/3-7\le-1/3$ and the claim again holds, 
or $n_3=4$, in which case $G$ has a bubble by~\refclaim{cl-no233233},
and hence $n_3+5n_2/3-7\le1/3\le(2b-1)/3$.
Finally, if $n_2=3$, then either $n_2\le1$, in which case
$n_3+5n_2/3-7\le-1$ and the claim holds; 
or $n_2=2$, in which case $G$ has a bubble by~\refclaim{cl-no22r}
and~\refclaim{cl-no232} and hence the claim holds;
or $n_2=3$, in which case $G$ has at least two bubbles by~\refclaim{cl-no22r}
and two applications of~\refclaim{cl-no232}, and hence the claim holds.
}

Let us now redistribute the charge according to the following rules: every face
distinct from $R$ sends $1/3$ to each incident vertex of degree two
and each incident internal vertex of degree three.  The face $T$ sends $1/3$ to each
face that shares an edge with it.  
The final charge of each vertex and of the faces $R$ and $T$
is clearly non-negative.
Since the sum of the final charges is equal to the sum of the initial charges,
it follows from~\refclaim{bubbles} that $G$ has a face $f$ of strictly negative final
charge. The face $f$ has length five; let $v_1,v_2,v_3,v_4,v_5$ be the
incident vertices in order.

If say $v_2$ were a vertex of degree two, then by \refclaim{cl-no22r}, 
$v_1$ and $v_3$ would be  vertices of $R$
of degree at least three, and hence $f$ would send no charge to them,
contrary to the fact that the final charge of $f$ is strictly negative.
It follows that all vertices of $f$ have degree at least three, and since the 
final charge of $f$ is negative, $f$ sends charge to at least four of them.  
Therefore, at least four of the vertices incident with $f$
are internal and have degree three.  
The fifth vertex has degree at least four
by \refclaim{cl-no53}: this is clear if it is internal, and otherwise
it has two neighbors on $R$ and two neighbors on $f$.
By \refclaim{cl-5faces} $f$ shares an edge with $T$.  
However, $f$ sends $1/3$ to each of its incident vertices of degree three
and nothing to the fifth vertex,
and receives $1/3$ from $T$; hence the final charge of $f$ is non-negative,
a contradiction.~\qed
\bigskip

We are now ready to prove Theorem~\ref{thm-6plus3}.
\bigskip

\noindent
{\bf Proof of Theorem~\ref{thm-6plus3}.}
Let $G,T$ and $\phi$ be as in Theorem~\ref{thm-6plus3}.
We may assume that $R$ bounds the outer face.  Let $p$ be any point of the
plane contained in the open disk bounded by $T$.
Let $G'$ be a minimal subgraph of $G$ such that $R$ is a subgraph of $G'$
and $\phi$ does not extend to a $3$-coloring of $G'$.
It follows that $G'$ is $R$-critical.
Note that $G'$ satisfies hypothesis ($*$) of  Theorem~\ref{thm-6plus3new}.
By Theorem~\ref{thm-6plus3new} the graph $G'$ is isomorphic to one of
the graphs depicted in Figure~\ref{fig-6plus3}.
If neither of the two outcomes of Theorem~\ref{thm-6plus3} holds,
then $\phi$ extends to a $3$-coloring of $G'$, a contradiction.~\qed

\def\JCTB{{\it J.~Combin.\ Theory Ser.\ B}}
\def\JGT{{\it J.~Graph Theory}}


\begin{thebibliography}{99}
\bibitem{aksenov} V.~A.~Aksionov,
On continuation of $3$-colouring of planar graphs (in Russian),
{\em Diskret. Anal.  Novosibirsk} {\bf 26} (1974), 3--19.

\bibitem{AlbHut} M.~Albertson and J.~Hutchinson,
The three excluded cases of Dirac's map-color theorem,
 Second International Conference on Combinatorial Mathematics
(New York, 1978),  pp.~7--17,
{\it Ann.\ New York Acad.\ Sci.} {\bf 319}, New York Acad.\ Sci.,
New York, 1979.

\bibitem{AppHak1} K. Appel and W. Haken,
Every planar map is four colorable, Part I: discharging,
{\it Illinois J. of Math.} {\bf 21} (1977), 429--490.

\bibitem{AppHakKoc} K. Appel, W. Haken and J. Koch,
Every planar map is four colorable, Part II: reducibility,
{\it Illinois J. of Math.} {\bf 21} (1977), 491--567.

\bibitem{AppHak86} K.~Appel and W.~Haken,
The four color proof suffices,
{\it The Mathematical Intelligencer \bf 8} (1986), 10--20.

\bibitem{AppHak89} K. Appel and W. Haken,
Every planar map is four colorable,
{\it Contemp. Math.} {\bf 98} (1989).

\bibitem{BonMur} J.~A.~Bondy and U.~S.~R.~Murty,
Graph Theory with Applications,
North-Holland, New York, Amsterdam, Oxford, 1976.

\bibitem{BorGleMonRas} O.~V.~Borodin, A.~N.~Glebov, M.~Montassier, A.~Raspaud,
Planar graphs without 5- and 7-cycles and without adjacent triangles are 3-colorable,
\JCTB\ {\bf99} (2009), 668--673.

\bibitem{ChePosStrThoYer} N.~Chenette, L.~Postle, N.~Streib, R.~Thomas and C.~Yerger,
Five-coloring graphs on the Klein bottle,
\JCTB\ {\bf 102} (2012), 1067--1098.

\bibitem{Dirmap} G. A. Dirac, Map color theorems,
{\it Canad.\ J.~Math.} {\bf 4} (1952), 480--490.

\bibitem{DvoKawTho} Z.~Dvo\v{r}\'ak, K.~Kawarabayashi and R.~Thomas,
Three-coloring triangle-free planar graphs in linear time,
{\it ACM Transactions on Algorithms} {\bf 7} (2011), article no. 41.

\bibitem{proof-druhy} Z.~Dvo\v{r}\'ak, D.~Kr\'al' and R.~Thomas,
Three-coloring triangle-free graphs on surfaces II. 
$4$-critical graphs in a disk, {\tt arXiv:1302.2158}.

\bibitem{proof-lincrit} Z.~Dvo\v{r}\'ak, D.~Kr\'al' and R.~Thomas,
Three-coloring triangle-free graphs on surfaces III. Graphs of girth five, 
{\tt arXiv:1402.4710}.

\bibitem{proof-havel} Z.~Dvo\v{r}\'ak, D.~Kr\'al' and R.~Thomas,
Three-coloring triangle-free graphs on surfaces V. Coloring planar graphs with distant anomalies,
{\tt arXiv:0911.0885}.

\bibitem{Fis} S.~Fisk, The nonexistence of colorings,
\JCTB\ {\bf 24} (1978), 247--248.

\bibitem{GarJoh} M.~R.~Garey and D.~S.~Johnson,
Computers and intractability. A guide
to the theory of NP-completeness, W. H. Freeman, San Francisco, 1979.

\bibitem{gimbel} J.~Gimbel and C.~Thomassen,
Coloring graphs with fixed genus and girth,
{\em Trans. Amer. Math. Soc.} {\bf 349} (1997), 4555--4564.

\bibitem{grotzsch} H.~Gr\"otzsch,
Ein Dreifarbensatz f\"ur dreikreisfreie Netze auf der Kugel,
{\em Wiss. Z. Martin-Luther-Univ. Halle-Wittenberg Math.-Natur. Reihe}
{\bf 8} (1959), 109--120.

\bibitem{grunbaum} B.~Gr\"unbaum,
Gr\"tzsch's theorem on $3$-colorings,
{\it Michigan Math. J.} {\bf 10} (1963), 303--310.

\bibitem{conj-havel} I.~Havel,
On a conjecture of Gr\"unbaum,
\JCTB\ {\bf 7} (1969), 184--186.

\bibitem{Heawood} P.~J.~Heawood, Map-color theorem,
{\it Quart.\ J.~Pure Appl.\ Math.} {\bf 24} (1890), 332--338.

\bibitem{KawKraKynLid} K.~Kawarabayashi, D.~Kral, J.~Kyn\v{c}l, and B.~Lidick\'y,
6-critical graphs on the Klein bottle,
{\it SIAM J.~Discrete Math.}  {\bf23}  (2008/09), 372--383.

\bibitem{Ringel} G.~Ringel, Map Color Theorem,
Springer-Verlag, Berlin, 1974.

\bibitem{RobSanSeyTho4CT}
N.~Robertson, D.~P.~Sanders, P.~D.~Seymour and R.~Thomas,
The four-colour theorem, \JCTB\ {\bf 70} (1997), 2-44.

\bibitem{thomwalls} R.~Thomas and B.~Walls,
Three-coloring Klein bottle graphs of girth five,
\JCTB\ {\bf  92}  (2004), 115--135.

\bibitem{Tho5torus} C. Thomassen, Five-coloring graphs on the torus,
\JCTB\ {\bf62} (1994), 11--33.

\bibitem{thom-torus} C.~Thomassen,
Gr\"otzsch's 3-Color Theorem and Its Counterparts for the Torus and the Projective Plane,
\JCTB\ {\bf 62} (1994), 268--279.

\bibitem{Tho3list} C.~Thomassen,
$3$-list coloring planar graphs of girth 5,
\JCTB\ {\bf64} (1995), 101--107.

\bibitem{ThoCritical} C. Thomassen, Color-critical graphs on a fixed surface,
\JCTB\ {\bf70} (1997), 67--100.

\bibitem{ThoShortlist} C. Thomassen,
A short list color proof of Grotzsch's theorem,
\JCTB\ {\bf88} (2003), 189--192.

\bibitem{thom-surf} C.~Thomassen,
The chromatic number of a graph of girth 5 on a fixed surface,
{\em J.~Combin. Theory Ser. B} {\bf 87} (2003), 38--71.

\bibitem{conj-stein} R.~ Steinberg,
The state of the three color problem. Quo Vadis, Graph Theory?,
{\em Ann. Discrete Math.} {\bf 55} (1993), 211--248.

\bibitem{WalPhD} B.~Walls, Coloring girth restricted graphs on surfaces,
Ph.D.\ dissertation, Georgia Institute of Technology, 1999.

\bibitem{xu} B.~Xu,
On $3$-colorable plane graphs without $5$- and $7$-cycles, {\em J.~Combin. Theory Ser. B} {\bf 96} (2006), 958--963.

\bibitem{Youngs} D.~A.~Youngs, $4$-chromatic projective graphs,
\JGT\ {\bf21} (1996), 219--227.

\end{thebibliography}
\end{document}